\documentclass[11pt]{article}

\usepackage[margin=1in]{geometry}
\usepackage[T1]{fontenc}
\usepackage[utf8]{inputenc}
\usepackage{charter}

\usepackage[protrusion=false]{microtype}
\usepackage{setspace}
\usepackage{amsmath,amssymb,amsthm,mathtools,bm}
\usepackage{booktabs,tabularx,array,threeparttable,longtable}
\usepackage{graphicx}
\graphicspath{{./}{figures/}}
\usepackage{subcaption}
\usepackage[dvipsnames]{xcolor}
\usepackage{enumitem}
\usepackage{fancyhdr}
\usepackage{lastpage}
\usepackage{placeins}
\usepackage[authoryear,round]{natbib}
\usepackage{hyperref}
\usepackage[nameinlink,capitalize,noabbrev]{cleveref}

\definecolor{navy}{HTML}{17365D}
\definecolor{teal}{HTML}{1B7F79}
\definecolor{lightblue}{HTML}{EDF4FA}
\definecolor{lightgray}{HTML}{F4F5F7}
\hypersetup{
  colorlinks=true,
  linkcolor=navy,
  citecolor=teal,
  urlcolor=navy,
  pdftitle={Efficient estimation and the cost of complete-case coarsening under monotone sequential MAR},
  pdfauthor={Keivan Bolouri}
}

\newtheorem{assumption}{Assumption}
\newtheorem{theorem}{Theorem}
\newtheorem{proposition}{Proposition}
\newtheorem{corollary}{Corollary}
\newtheorem{lemma}{Lemma}
\theoremstyle{definition}
\newtheorem{remark}{Remark}
\crefname{assumption}{Assumption}{Assumptions}
\Crefname{assumption}{Assumption}{Assumptions}
\crefname{theorem}{Theorem}{Theorems}
\crefname{proposition}{Proposition}{Propositions}
\crefname{corollary}{Corollary}{Corollaries}
\crefname{lemma}{Lemma}{Lemmas}

\newcommand{\E}{\mathbb{E}}
\newcommand{\Pp}{\mathbb{P}}
\newcommand{\Var}{\operatorname{Var}}

\newcommand{\indep}{\perp\!\!\!\perp}
\newcommand{\expit}{\operatorname{expit}}
\newcommand{\norm}[1]{\left\lVert #1\right\rVert_2}
\newcommand{\1}{\mathbf{1}}
\newcommand{\cc}{\mathrm{cc}}
\newcommand{\seq}{\mathrm{seq}}
\newcommand{\F}{\mathrm{F}}

\title{\vspace{-1.2cm}\textbf{Efficient estimation and the cost of complete-case\\ coarsening under monotone sequential MAR}}
\author{Keivan Bolouri\\
\small Department of Statistics and Data Science, University of California, Los Angeles,\\
\small Los Angeles, California, USA\\[2pt]
\small Correspondence: \href{mailto:keivanbolouri@ucla.edu}{keivanbolouri@ucla.edu}}
\date{September 10, 2026}

\begin{document}
\maketitle

\begin{spacing}{1.0}
\begin{abstract}
Complete-case coarsening discards observed confounder values from partially complete records. We study its consequences for average treatment effect estimation with two ordered, partially observed confounders under monotone sequential missing at random. We specialize the standard coarsening-at-random transformation to the causal influence function, establish the canonical gradient, and give an exact drift identity for a cross-fitted estimator with sequential multiple robustness. In the submodel where both sequential and complete-case missing-at-random assumptions hold, we express the efficiency loss from coarsening as a nonnegative expectation involving two iterated projections. The gain is strict when the intermediate confounder supplies residual information where second-stage missingness occurs. Oracle simulations and deterministic quadrature illustrate this efficiency comparison. When second-stage response depends on the intermediate confounder, the comparison instead concerns identification: coarsening can introduce persistent bias. Estimated-nuisance simulations include a bounded-propensity design and a Gaussian stress design. The latter exhibits substantial interval undercoverage and can reverse the finite-sample precision ordering, qualifying the practical interpretation of the efficiency bound.
\end{abstract}



\end{spacing}
\clearpage

\section{Introduction}

Observational studies frequently combine two statistical difficulties: treatment is confounded, and one or more baseline confounders are missing. Complete-case analysis is generally inefficient and can be biased. Multiple imputation, inverse-probability weighting, augmented estimators, generalized raking, and targeted learning address parts of this problem, but their validity depends on how the causal and missingness assumptions fit the observed-data structure \citep{bang2005,williamson2012,seaman2014,levis2025,benz2025,williamson2026,wen2026}.

\citet{levis2025} developed robust and efficient ATE estimators under complete-case missing at random (CCMAR). If $L_c$ denotes fully observed confounders, $L_p$ the vector of partially observed confounders, and $S$ the indicator that every component of $L_p$ is observed, their assumption is
\[
S\indep L_p\mid (L_c,A,Y).
\]
The reduction to $S$ avoids modeling a possibly nonmonotone collection of response patterns. It also means that subjects with only part of $L_p$ observed do not directly contribute their observed confounder values to the final estimating function. Levis and colleagues identify the use of such partial records under alternative assumptions as an important direction for future research.

\citet{benz2025} provide the applied counterpart to that development, comparing
the robust estimators of \citet{levis2025} against ad hoc alternatives---%
sequential imputation followed by outcome regression, or by inverse-probability
weighting---across a range of missing-confounder designs. They report that no
single estimator is uniformly best and that the simpler alternatives are
adequate in several settings. Their study motivates a practical distinction: an efficiency ordering for influence functions need not translate into a finite-sample precision ordering when nuisance functions are estimated. We study that distinction in \cref{app:crossfit}.

This paper studies the transparent special case in which the components of $L_p$ have a scientifically meaningful order and are observed monotonically. We write
\[
L_c=W,\qquad L_p=(L_1,L_2),
\]
so the notation $L_1,L_2$ refines---rather than replaces---the earlier $L_c,L_p$ notation. The possible patterns are: neither partially observed confounder is available; $L_1$ alone is available; or both are available. Sequential MAR permits the probability of observing $L_2$ to depend on the already observed $L_1$. This model is useful when, for example, an inexpensive screening variable precedes a more burdensome laboratory measure, or when data acquisition follows an ordered protocol.

Sequential inverse-probability augmentation and efficient estimation with monotone missingness are established ideas \citep{rrz1994,gill1997,tsiatis2006,barnwell2025}. Accordingly, our contribution is not the generic sequential augmentation formula. It is the following causal and comparative analysis:
\begin{enumerate}[leftmargin=*,itemsep=3pt]
\item We specialize the monotone sequential construction to the full-data influence function for a point-exposure ATE with missing confounders, allowing response probabilities to depend on the fully observed treatment and outcome.
\item We give a self-contained observed-likelihood and tangent-space proof of the efficient influence function, followed by an exact remainder identity and a cross-fitted estimator.
\item On the submodel where sequential MAR and CCMAR are simultaneously valid, we derive the exact nonnegative efficiency loss from collapsing the intermediate pattern.
\item We distinguish this efficiency comparison from the case in which second-stage response depends on $L_1$. There, CCMAR generally fails, so the comparison concerns identification, bias, and coverage rather than two efficiency bounds for the same model.
\end{enumerate}

The final distinction is practically important. A statement that the pattern-aware method is ``more efficient'' is justified only under common validity. Outside that submodel, better mean-squared error may result primarily from removing bias under a different identifying assumption. Sequential MAR and unrestricted CCMAR are not generally nested models; within the sequential model, \cref{ass:common} defines their common-validity submodel.

\section{Observed data, causal target, and assumptions}

\subsection{Notation and observation patterns}

For one individual, let
\[
Z_0=(W,A,Y),\qquad Z_1=(Z_0,L_1),\qquad Z_2=(Z_1,L_2),
\]
where $W$ is a vector of fully observed baseline confounders, $A\in\{0,1\}$ is a fully observed point exposure, $Y$ is a fully observed outcome, and $L_1,L_2$ are partially observed baseline confounders. Let $R_1$ indicate that $L_1$ is observed. Among records with $R_1=1$, let $C_2$ indicate that $L_2$ is observed, and define $R_2=R_1C_2$. Thus
\[
(R_1,R_2)\in\{(0,0),(1,0),(1,1)\};
\]
the $L_2$-only pattern $(0,1)$ is excluded. The observed datum is
\[
O=(Z_0,R_1,R_1L_1,R_2,R_2L_2).
\]

\begin{table}[ht]
\centering
\caption{Relationship between the original CCMAR notation and the ordered notation used here.}
\label{tab:notation}
\begin{tabularx}{0.92\textwidth}{@{}lXX@{}}
\toprule
Object & CCMAR notation & Sequential notation \\
\midrule
Always observed confounders & $L_c$ & $W$ \\
Partially observed confounders & $L_p$ & $(L_1,L_2)$ \\
Complete-case indicator & $S$ & $R_2=R_1C_2$ \\
Intermediate partial record & discarded by coarsening & $R_1=1,R_2=0$ \\
\bottomrule
\end{tabularx}
\end{table}

Let $V=(W,L_1,L_2)$. The causal estimand is
\[
\psi=\E\{Y(1)-Y(0)\}.
\]

\begin{assumption}[Causal identification]\label{ass:causal}
For $a\in\{0,1\}$: (i) consistency, $Y=Y(A)$; (ii) conditional exchangeability, $Y(a)\indep A\mid V$; and (iii) treatment positivity, $\epsilon_e<e(V)<1-\epsilon_e$ almost surely for some $\epsilon_e>0$, where $e(v)=\Pp(A=1\mid V=v)$.
\end{assumption}

Under \cref{ass:causal}, writing $\mu_a(v)=\E(Y\mid A=a,V=v)$,
\begin{equation}\label{eq:ate-gformula}
\psi=\E\{\mu_1(V)-\mu_0(V)\}.
\end{equation}

\subsection{Sequential missing at random}

\begin{assumption}[Monotone sequential MAR]\label{ass:smar}
The response process satisfies
\begin{align}
R_1&\indep (L_1,L_2)\mid Z_0,\label{eq:mar1}\\
C_2&\indep L_2\mid (Z_1,R_1=1).\label{eq:mar2}
\end{align}
Let
\[
\pi_1(Z_0)=\Pp(R_1=1\mid Z_0),\qquad
\pi_2(Z_1)=\Pp(C_2=1\mid R_1=1,Z_1),
\]
and suppose there is $\epsilon_\pi>0$ such that $\pi_1(Z_0)\ge\epsilon_\pi$ and
$\pi_2(Z_1)\ge\epsilon_\pi$ almost surely.
\end{assumption}

Treatment and outcome are included in $Z_0$, so both response probabilities may depend on $A$ and $Y$. These are MAR/CAR restrictions because the hazards depend only on information observed at the corresponding stage. They are not outcome-independent missing-not-at-random restrictions.

\begin{proposition}[Identification of the full-data law]\label{prop:identification}
Under \cref{ass:smar}, the full-data density is identified by
\begin{equation}\label{eq:full-law-identification}
p(z_0,l_1,l_2)
=p(z_0)\,p(l_1\mid z_0,R_1=1)\,
p(l_2\mid z_1,R_1=1,C_2=1).
\end{equation}
Consequently, under \cref{ass:causal}, $\psi$ in \cref{eq:ate-gformula} is identified from the observed-data law.
\end{proposition}

\begin{proof}
By \cref{eq:mar1}, $p(l_1,l_2\mid z_0,R_1=1)=p(l_1,l_2\mid z_0)$, hence $p(l_1\mid z_0,R_1=1)=p(l_1\mid z_0)$. By \cref{eq:mar2}, among $R_1=1$, $p(l_2\mid z_1,C_2=1)=p(l_2\mid z_1)$. Combining these conditional densities with $p(z_0)$ gives \cref{eq:full-law-identification}. The g-formula is a functional of the identified full-data law.
\end{proof}

\Cref{eq:full-law-identification} is a retrospective factorization because $A$ and $Y$ occur in $Z_0$. This causes no contradiction with the causal factorization: any identified joint law can also be refactored as $p(v)p(a\mid v)p(y\mid a,v)$, after which \cref{ass:causal} gives the causal interpretation.

\section{Efficient influence function under sequential MAR}

\subsection{Full-data gradient}

The efficient influence function for the ATE in the nonparametric full-data model is
\begin{equation}\label{eq:full-eif}
D_{\F}(Z_2)=G(Z_2)-\psi,
\end{equation}
where
\begin{equation}\label{eq:G}
G(Z_2)=\mu_1(V)-\mu_0(V)
+\frac{A}{e(V)}\{Y-\mu_1(V)\}
-\frac{1-A}{1-e(V)}\{Y-\mu_0(V)\}.
\end{equation}
Define the iterated projections
\begin{equation}\label{eq:Qs}
Q_1(Z_1)=\E\{G(Z_2)\mid Z_1\},\qquad
Q_0(Z_0)=\E\{G(Z_2)\mid Z_0\}.
\end{equation}

Throughout, $\mathcal M$ denotes the observed-data model: the set of laws of
$O$ generated by an unrestricted full-data law $p(z_0,l_1,l_2)$ together with
unrestricted response probabilities $\pi_1,\pi_2$ satisfying \cref{ass:smar}
and its positivity bound. \Cref{lem:tangent} in \cref{app:tangent} shows that
the tangent space of $\mathcal M$ is all of $L_2^0(P_O)$, the mean-zero square-integrable
functions of $O$. The model is therefore nonparametric in the observed-data
sense, a pathwise differentiable parameter has exactly one gradient, and that
gradient is its efficient influence function.

\begin{theorem}[Canonical gradient]\label{thm:eif}
Suppose \cref{ass:causal,ass:smar} hold, $D_{\F}\in L_2(P)$, and the response
probabilities satisfy the two-sided bounds of \cref{lem:tangent}. In the
nonparametric observed-data model $\mathcal M$ induced by monotone sequential
MAR, the efficient influence function for $\psi$ is
\begin{equation}\label{eq:seq-eif}
\boxed{
D_{\seq}(O)=Q_0-\psi
+\frac{R_1}{\pi_1}(Q_1-Q_0)
+\frac{R_2}{\pi_1\pi_2}(G-Q_1).}
\end{equation}
All functions in a term are evaluated only when the variables required for that term are observed.
\end{theorem}

\begin{proof}
We give the observed-likelihood argument. Write
\[
p_0(z_0)=p(z_0),\quad p_1(l_1\mid z_0)=p(l_1\mid z_0),\quad
p_2(l_2\mid z_1)=p(l_2\mid z_1).
\]
Under \cref{ass:smar}, the observed-data likelihood contribution factorizes as
\begin{align}
p_O(o)
={}&p_0(z_0)
\{1-\pi_1(z_0)\}^{1-r_1}
\{\pi_1(z_0)p_1(l_1\mid z_0)\}^{r_1}\notag\\
&\times\{1-\pi_2(z_1)\}^{r_1-r_2}
\{\pi_2(z_1)p_2(l_2\mid z_1)\}^{r_2}.
\label{eq:observed-likelihood}
\end{align}
For regular parametric submodels, the tangent space is the closure of sums of the mutually orthogonal components
\begin{align*}
\mathcal T_0&=\{s_0(Z_0):\E(s_0)=0\},\\
\mathcal T_1&=\{R_1s_1(Z_1):\E(s_1\mid Z_0)=0\},\\
\mathcal T_2&=\{R_2s_2(Z_2):\E(s_2\mid Z_1)=0\},\\
\mathcal T_{\pi_1}&=\{(R_1-\pi_1)h_1(Z_0)\},\\
\mathcal T_{\pi_2}&=\{(R_2-R_1\pi_2)h_2(Z_1)\}.
\end{align*}
Let $q_1=\E(D_{\F}\mid Z_1)=Q_1-\psi$ and
$q_0=\E(D_{\F}\mid Z_0)=Q_0-\psi$. Candidate \cref{eq:seq-eif} can be written
\begin{equation}\label{eq:eif-centered}
D_{\seq}=q_0+\frac{R_1}{\pi_1}(q_1-q_0)
+\frac{R_2}{\pi_1\pi_2}(D_{\F}-q_1).
\end{equation}
\Cref{lem:tangent} shows that these five components are mutually orthogonal
and that their closed linear span is the whole of $L_2^0(P_O)$; every tangent
direction is therefore a sum of the five types, which is what makes the
verification below exhaustive rather than merely suggestive.

The three terms in \cref{eq:eif-centered} belong respectively to
$\mathcal T_0,\mathcal T_1,$ and $\mathcal T_2$, because the required conditional
means vanish by iterated expectation; the same argument gives
$\E(D_{\seq})=0$. Thus $D_{\seq}$ belongs to the observed-data tangent space and is
orthogonal to the two response-mechanism tangent components.

It remains to verify that $D_{\seq}$ represents the pathwise derivative. For $s_0\in\mathcal T_0$,
\[
\E(D_{\seq}s_0)=\E(q_0s_0)=\E(D_{\F}s_0).
\]
For a full-law score $s_1$ satisfying $\E(s_1\mid Z_0)=0$, the observed score is $R_1s_1$, and conditional expectation over $R_1,C_2$ gives
\[
\E(D_{\seq}R_1s_1)=\E\{(q_1-q_0)s_1\}=\E(D_{\F}s_1).
\]
The last equality uses $\E(s_1\mid Z_0)=0$ and $q_1=\E(D_{\F}\mid Z_1)$. Similarly, for $s_2$ satisfying $\E(s_2\mid Z_1)=0$, the observed score is $R_2s_2$ and
\[
\E(D_{\seq}R_2s_2)=\E\{(D_{\F}-q_1)s_2\}=\E(D_{\F}s_2).
\]
Finally, $\psi$ does not depend on $\pi_1$ or $\pi_2$, and the inner products of
$D_{\seq}$ with their score components are zero; this is verified in
\cref{app:scores}. By \cref{lem:tangent} the five families of scores exhaust the
tangent space, so \cref{eq:eif-centered} represents the derivative along every
tangent direction and $D_{\seq}$ is a gradient. Because it also lies in the
tangent space, it is the canonical gradient, and it is the unique influence
function of $\psi$ in $\mathcal M$; in particular it has minimum variance among
regular influence functions \citep{tsiatis2006}.
\end{proof}

\begin{remark}
The proof is a two-stage specialization of general monotone-MAR/CAR theory \citep{rrz1994,gill1997,barnwell2025}. Its purpose is to establish rigorously that the causal pseudo-outcome in \cref{eq:G}, after sequential projection, yields the efficient observed-data gradient for this particular ATE problem.
\end{remark}

\section{Estimation, robustness, and inference}

\subsection{Cross-fitted one-step estimator}

Let $O_1,\ldots,O_n$ be independent and identically distributed observations from $P_O$, and let $\eta=(e,\mu_0,\mu_1,\pi_1,\pi_2,Q_0,Q_1)$. With nuisance estimates trained outside observation $i$'s validation fold, define
\begin{equation}\label{eq:Hhat}
\widehat H_i=\widehat Q_{0,i}
+\frac{R_{1i}}{\widehat\pi_{1,i}}(\widehat Q_{1,i}-\widehat Q_{0,i})
+\frac{R_{2i}}{\widehat\pi_{1,i}\widehat\pi_{2,i}}
(\widehat G_i-\widehat Q_{1,i}),
\end{equation}
and
\begin{equation}\label{eq:estimator}
\widehat\psi_{\seq}=\frac{1}{n}\sum_{i=1}^n\widehat H_i.
\end{equation}
An estimated standard error is
\begin{equation}\label{eq:se}
\widehat{\mathrm{se}}(\widehat\psi_{\seq})
=\left[\frac{1}{n(n-1)}\sum_{i=1}^n
(\widehat H_i-\widehat\psi_{\seq})^2\right]^{1/2}.
\end{equation}

One practical training sequence is: fit $\pi_1$ to all training records; fit $\pi_2$ among $R_1=1$ records; fit the causal regressions on complete records with inverse-response weights to reconstruct the full law; regress $\widehat G$ on $Z_1$ among complete records to obtain $Q_1$; and regress the stage-two augmented pseudo-outcome on $Z_0$ among $R_1=1$ records to obtain $Q_0$. Other learners are possible, provided they estimate the same nuisance functions and the positivity and convergence conditions below are respected.

\subsection{Exact drift and sequential multiple robustness}

For arbitrary fixed functions $\bar G,\bar Q_0,\bar Q_1,\bar\pi_1,\bar\pi_2$ with positive response probabilities and finite displayed expectations, define
\[
Q_1^{\bar G}=\E(\bar G\mid Z_1),\qquad
Q_0^{\bar G}=\E(\bar G\mid Z_0),
\]
and let $H(\bar\eta)$ denote \cref{eq:Hhat} with bars in place of hats.

\begin{proposition}[Exact missingness-layer drift]\label{prop:drift}
Let all expectations be taken under the true law $P$, and let $Q_0^{\bar G}$ and
$Q_1^{\bar G}$ denote true conditional expectations of the working pseudo-outcome
$\bar G$. Under \cref{ass:smar} alone, and with no requirement that any working
nuisance be correct,
\begin{align}
\E\{H(\bar\eta)\}-\E(\bar G)
={}&\E\!\left[\left(1-\frac{\pi_1}{\bar\pi_1}\right)
(\bar Q_0-Q_0^{\bar G})\right]\notag\\
&+\E\!\left[\frac{\pi_1}{\bar\pi_1}
\left(1-\frac{\pi_2}{\bar\pi_2}\right)
(\bar Q_1-Q_1^{\bar G})\right].
\label{eq:missing-drift}
\end{align}
If $\bar G$ is formed from $(\bar e,\bar\mu_0,\bar\mu_1)$ as in \cref{eq:G}, then
\begin{align}
\E(\bar G)-\psi
=\E\!\left[
\frac{\bar e-e}{\bar e}(\bar\mu_1-\mu_1)
+\frac{\bar e-e}{1-\bar e}(\bar\mu_0-\mu_0)
\right].
\label{eq:causal-drift}
\end{align}
\end{proposition}

\begin{proof}
Taking expectation of $H(\bar\eta)$ over the response indicators conditional on the full data replaces $R_1$ by $\pi_1$ and $R_2$ by $\pi_1\pi_2$. Subtracting $\E(\bar G)$ and adding and subtracting $Q_j^{\bar G}$ yields \cref{eq:missing-drift}; the terms involving $Q_j^{\bar G}$ reduce to conditional expectations of $\bar G$. For \cref{eq:causal-drift}, condition the barred version of \cref{eq:G} on $V$, use $\E(A\mid V)=e(V)$ and $\E(Y\mid A=a,V)=\mu_a(V)$, then collect products of the treatment- and outcome-regression errors.
\end{proof}

\begin{corollary}[Sequential multiple robustness]\label{cor:robustness}
The estimating equation is unbiased for $\psi$ under every nuisance configuration satisfying all three clauses:
\begin{enumerate}[label=(\roman*),leftmargin=*]
\item $\bar e=e$ or $(\bar\mu_0,\bar\mu_1)=(\mu_0,\mu_1)$;
\item $\bar\pi_1=\pi_1$ or $\bar Q_0=Q_0^{\bar G}$;
\item $\bar\pi_2=\pi_2$ or $\bar Q_1=Q_1^{\bar G}$.
\end{enumerate}
This gives up to $2\times2\times2=8$ sufficient configurations. The result is a layered or sequential form of multiple robustness. The alternatives concern the true projections of the working pseudo-outcome; a particular regression specification is not guaranteed to attain each configuration, and an arbitrary subset of nuisance models cannot be misspecified without restriction.
\end{corollary}

\subsection{Asymptotic normality}

\begin{theorem}[Cross-fitted asymptotic linearity]\label{thm:asymptotics}
Use a fixed number of folds. Suppose, uniformly across folds: (i) fitted observation probabilities and treatment probabilities are bounded away from zero (and the treatment probability from one); (ii) $\norm{\widehat H-H(\eta)}=o_p(1)$ and the influence functions have uniformly integrable $2+\delta$ moments for some $\delta>0$; and (iii)
\begin{align*}
&\norm{\widehat e-e}
\{\norm{\widehat\mu_1-\mu_1}+\norm{\widehat\mu_0-\mu_0}\}
=o_p(n^{-1/2}),\\
&\norm{\widehat\pi_1-\pi_1}\,
\norm{\widehat Q_0-Q_0^{\widehat G}}=o_p(n^{-1/2}),\\
&\norm{\widehat\pi_2-\pi_2}\,
\norm{\widehat Q_1-Q_1^{\widehat G}}=o_p(n^{-1/2}).
\end{align*}
Then
\[
\sqrt n(\widehat\psi_{\seq}-\psi)
=\frac{1}{\sqrt n}\sum_{i=1}^nD_{\seq}(O_i)+o_p(1)
\rightsquigarrow N\{0,\Var(D_{\seq})\}.
\]
The variance estimator in \cref{eq:se} is consistent, and the estimator attains the semiparametric efficiency bound.
\end{theorem}

\begin{proof}
Conditional on the training samples, cross-fitting makes each validation-fold nuisance fit fixed. Decompose $\widehat\psi-\psi$ into the empirical mean of the true influence function, an empirical-process difference on the validation fold, and the conditional drift. The $L_2$ convergence in (ii) makes the second term $o_p(n^{-1/2})$ without Donsker conditions. Applying \cref{prop:drift} fold by fold and Cauchy--Schwarz bounds the drift by the three products in (iii), up to constants from positivity. Thus the remainder is $o_p(n^{-1/2})$. The central limit theorem applies to the first term. Consistency of the empirical second moment follows from (ii) and uniform integrability. Efficiency follows from \cref{thm:eif}. This is the standard cross-fitting argument \citep{chernozhukov2018} specialized to \cref{eq:seq-eif}.
\end{proof}

The conditions in (iii) constrain products of errors, not individual rates, so
one factor may converge slowly if its partner converges quickly. The familiar
$n^{-1/4}$ heuristic needs care: two factors that are exactly $O_p(n^{-1/4})$
give a product that is $O_p(n^{-1/2})$, which does \emph{not} satisfy (iii). A
correct simple sufficient condition is that each factor be $o_p(n^{-1/4})$, or
more generally that the two rates $r_1,r_2$ satisfy $r_1r_2=o(n^{-1/2})$. Very
small response or treatment probabilities remain a practical threat even when
formal positivity holds, motivating weight diagnostics, truncation sensitivity
analyses, and estimators targeted for finite-sample stability.

\section{Exact cost of complete-case coarsening}

Let $S=R_2$ and suppose the analyst replaces $(R_1,R_2)$ by the single
complete-case indicator $S$, so that the working data are
$O^{\cc}=(Z_0,S,SL_1,SL_2)$. Such an analyst cannot use $\pi_2$, which is a
function of $L_1$ and so is unavailable exactly on the records that $S$
discards; the estimable response functional is instead
\begin{equation}\label{eq:rho}
\rho(Z_0)=\Pp(S=1\mid Z_0)=\pi_1(Z_0)\,\E\{\pi_2(Z_1)\mid Z_0\}.
\end{equation}
The next restriction creates a submodel in which both approaches identify the
same target and, in addition, $\rho=\pi_1\pi_2$.

\begin{assumption}[Common-validity submodel]\label{ass:common}
The second-stage response probability does not depend on the intermediate confounder after $Z_0$:
\[
\pi_2(Z_1)=\pi_2(Z_0).
\]
\end{assumption}

Then $S\indep(L_1,L_2)\mid Z_0$ and $\rho(Z_0)=\pi_1(Z_0)\pi_2(Z_0)$, so CCMAR
holds. The efficient influence function based only on $O^{\cc}$ is
\begin{equation}\label{eq:cc-eif}
D_{\cc}(O^{\cc})=Q_0-\psi+
\frac{S}{\pi_1\pi_2}(G-Q_0).
\end{equation}

\begin{theorem}[Efficiency loss from collapsing the pattern]\label{thm:gap}
Under \cref{ass:causal,ass:smar,ass:common}, with $D_{\F}\in L_2(P)$,
\begin{equation}\label{eq:difference}
D_{\cc}-D_{\seq}
=\frac{R_1}{\pi_1}\left(\frac{C_2}{\pi_2}-1\right)(Q_1-Q_0),
\end{equation}
and this difference is orthogonal to $D_{\seq}$. Consequently,
\begin{equation}\label{eq:variance-gap}
\boxed{
\Var(D_{\cc})-\Var(D_{\seq})
=\E\!\left[
\frac{1-\pi_2(Z_0)}{\pi_1(Z_0)\pi_2(Z_0)}
\{Q_1(Z_1)-Q_0(Z_0)\}^2
\right]\ge0.}
\end{equation}
\end{theorem}

\begin{proof}
Expanding \cref{eq:cc-eif} and subtracting \cref{eq:seq-eif} gives \cref{eq:difference}. Let $\Delta=Q_1-Q_0$ and $\varepsilon=G-Q_1$. Conditional on $Z_1$, $\E(\varepsilon\mid Z_1)=0$; conditional on $Z_1,R_1=1$, $\E(C_2/\pi_2-1)=0$. Multiplying \cref{eq:difference} by the three orthogonal increments in \cref{eq:seq-eif} and applying these conditional-mean properties shows
$\E\{(D_{\cc}-D_{\seq})D_{\seq}\}=0$. Finally,
\begin{align*}
\E\{(D_{\cc}-D_{\seq})^2\mid Z_1\}
&=\frac{1-\pi_2(Z_0)}{\pi_1(Z_0)\pi_2(Z_0)}\Delta^2,
\end{align*}
because $R_1$ and $C_2$ are Bernoulli with the stated sequential probabilities. Pythagoras gives \cref{eq:variance-gap}.
\end{proof}

\begin{corollary}[When the gain is strict]\label{cor:strict}
The pattern-aware efficiency bound is strictly smaller if and only if
\[
\Pp\!\left\{\pi_2(Z_0)<1\ \text{and}\ Q_1(Z_1)\ne Q_0(Z_0)\right\}>0.
\]
Equivalently, strict improvement requires second-stage missingness and residual predictive information in $L_1$ about the full-data influence function after conditioning on $Z_0$.
\end{corollary}

When $\pi_2$ depends on $L_1$, the complete-case probability conditional on $Z_0$
averages over an unobserved determinant of selection, and $S\indep(L_1,L_2)\mid Z_0$
generally fails. Then \cref{eq:variance-gap} is not an efficiency comparison:
the coarsened estimator may target the wrong observed-data functional. This is
the distinction between an \emph{information loss} under common validity and an
\emph{identification failure} under sequential MAR alone.

Outside \cref{ass:common} the coarsened analyst must also replace the two
quantities in \cref{eq:cc-eif} that are no longer estimable from $O^{\cc}$:
$\pi_1\pi_2$ becomes $\rho(Z_0)$ of \cref{eq:rho}, and $Q_0(Z_0)=\E(G\mid Z_0)$
becomes the complete-case projection
\[
Q_0^{\cc}(Z_0)=\E(G\mid Z_0,S=1)
=\frac{\E\{\pi_2(Z_1)Q_1(Z_1)\mid Z_0\}}{\E\{\pi_2(Z_1)\mid Z_0\}},
\]
which coincides with $Q_0$ under \cref{ass:common} and differs from it otherwise.
Neither substitution generally restores consistency. With the true causal pseudo-outcome and $\rho$, either choice of augmentation has population mean $\E(Q_0^{\cc})$: conditional expectation given $Z_0$ cancels the augmentation. \Cref{sec:sim} reports both this bias calculation and oracle simulations. The estimated-nuisance coarsened procedure also fits its causal regressions using complete records weighted by $1/\widehat\rho$; outside common validity, those regressions need not estimate the full-law causal nuisances.

\section{Simulation studies}\label{sec:sim}

\subsection{Gaussian stress design and response mechanisms}

The first design isolates the missingness comparison with oracle nuisance functions, then examines the effect of estimating those functions. There are no fully observed baseline covariates $W$. The full-data law is specified retrospectively:
\begin{align*}
A&\sim\mathrm{Bernoulli}(0.5),\\
Y\mid A=0&\sim\mathrm{Beta}(2,4),\qquad
Y\mid A=1\sim\mathrm{Beta}(4,2),\\
L_1\mid A,Y&\sim\mathrm{Bernoulli}\{\expit(-0.6+0.5A+0.25Y+0.1AY)\},\\
L_2\mid A,Y,L_1&\sim N(A+Y+2.5L_1Y,1.25^2).
\end{align*}
Refactorization of this law yields $e(V)$ and $\mu_a(V)$ by numerical quadrature. The ATE is $\psi=0.2558505161$ to ten decimal places; all summaries use the unrounded computed value. Stage-one response is
\[
\pi_1=\expit(2+0.1A+0.1Y).
\]
The two stage-two mechanisms are
\[
\pi_2=\begin{cases}
\expit(-1.1+0.1A+0.1Y),&\text{common validity},\\
\expit(-1.9+0.1A+0.1Y+2.2L_1),&\text{sequential MAR only}.
\end{cases}
\]
The first permits an efficiency comparison and the second illustrates identification failure after coarsening. The oracle experiment uses 2,000 samples of size 2,500 per mechanism. All three estimators are means of their influence-function pseudo-outcomes, with normal intervals based on the empirical variance within each sample. The coarsened oracle procedure uses the true causal functions and $Q_0$ but weights by $1/\rho(Z_0)$; the supplement also reports the alternative augmentation $Q_0^{\cc}$. Neither oracle version is a feasible causal estimator outside common validity.

This Gaussian design does not satisfy the uniform treatment-positivity condition in \cref{ass:causal}. The unbounded support of $L_2$ permits treatment probabilities arbitrarily close to zero or one. The weaker moment conditions $\E(1/e)=4.04$ and $\E\{1/(1-e)\}=5.65$ nevertheless give a finite full-data efficiency bound because $Y$ is bounded. Deterministic quadrature gives $\Var(G)=0.27924$ and $\E(G^4)=105.28$. These features make this a stress design for learning the causal nuisance functions; it is not an empirical verification of \cref{thm:asymptotics}.

\subsection{Oracle comparison}

\begin{table}[htbp]
\centering
\begin{threeparttable}
\caption{Oracle simulation in the Gaussian stress design ($n=2{,}500$, 2,000 replicates per scenario).}
\label{tab:main-results}
\small
\setlength{\tabcolsep}{4pt}
\begin{tabular}{@{}llrrrr@{}}
\toprule
Scenario & Estimator & Bias & SD & RMSE & Coverage \\
\midrule
Common validity & Full data & $-0.00005$ & $0.01043$ & $0.01043$ & 94.95\% \\
 & Coarsened & $-0.00004$ & $0.01775$ & $0.01775$ & 95.10\% \\
 & Sequential & $-0.00003$ & $0.01723$ & $0.01722$ & 95.55\% \\
\addlinespace
Sequential MAR only & Full data & $0.00013$ & $0.01052$ & $0.01051$ & 95.05\% \\
 & Coarsened & $-0.03301$ & $0.01848$ & $0.03783$ & 45.70\% \\
 & Sequential & $-0.00013$ & $0.01626$ & $0.01626$ & 94.65\% \\
\addlinespace
\bottomrule
\end{tabular}
\begin{tablenotes}[flushleft]\footnotesize
\item The target ATE is $0.2558505161$. The coarsened column uses true $G$, $Q_0$, and $\rho$. Bias and coverage Monte Carlo standard errors are recorded in the supplement; coverage MCSE is at most 1.12 percentage points.
\end{tablenotes}
\end{threeparttable}
\end{table}

The centered empirical variance reduction under common validity was 5.83\%, with paired-bootstrap Monte Carlo 95\% interval 4.03\%--7.57\% (10,000 resamples). This statistic uses deviations from the replicate means, rather than squared deviations from the true ATE, which would estimate a mean-squared-error reduction.

Deterministic product quadrature gives influence-function variances 0.79524 and 0.76224 for the coarsened and sequential procedures, respectively. Their difference is 0.03300, corresponding to a 4.15\% reduction. Doubling the quadrature resolutions changes each variance by less than $10^{-8}$, and the numerical residual in \cref{eq:variance-gap} is below $10^{-12}$. These are converged numerical evaluations, not symbolic integrations. The 4.15\% value lies within the simulation interval; the interval's width also shows why a finite Monte Carlo variance ratio should not be equated with the efficiency bound.

In the sequential-only scenario, the coarsened oracle bias was $-0.03301$ with 45.70\% coverage, whereas the sequential bias was $-0.00013$ with 94.65\% coverage. Quadrature gives population coarsened bias $-0.03282$. Replacing $Q_0$ by $Q_0^{\cc}$ gives simulated bias $-0.03302$; the two population biases coincide, as shown above. This is an identification comparison.

\begin{figure}[htbp]
\centering
\includegraphics[width=\textwidth]{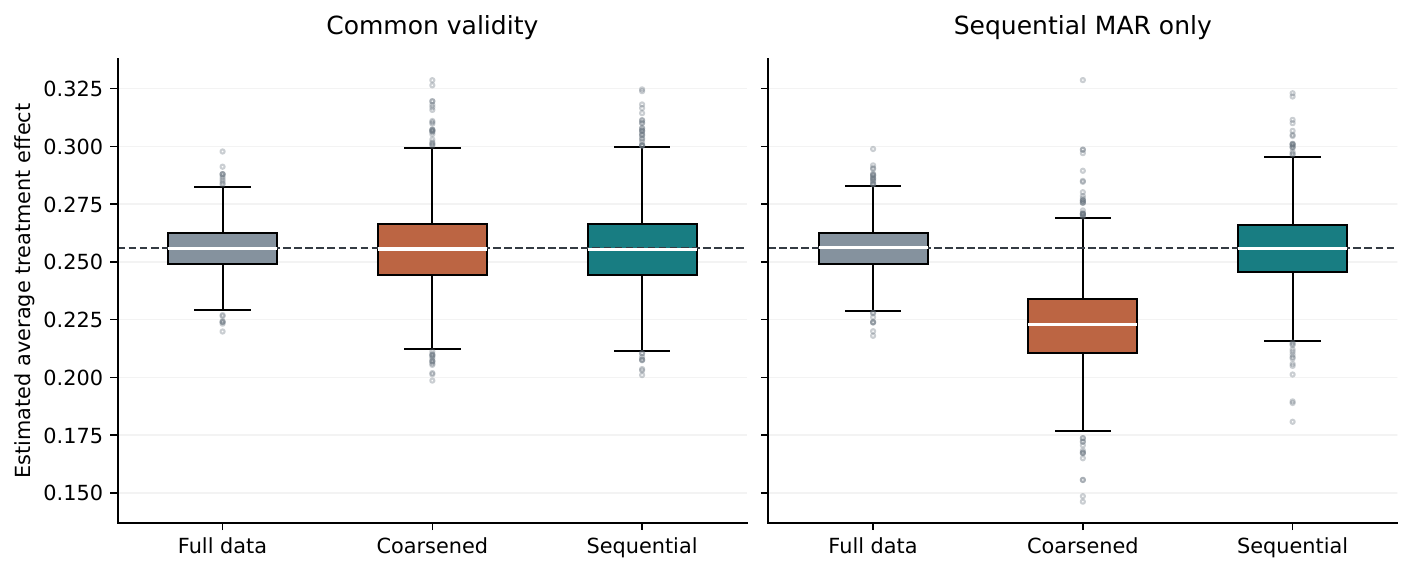}
\caption{Oracle estimates from the same replicates used in \cref{tab:main-results}. Boxes show quartiles, whiskers extend to 1.5 times the interquartile range, and points beyond the whiskers are shown. Dashed lines mark the true ATE.}
\label{fig:main-boxplot}
\end{figure}

The population pattern proportions obtained by quadrature are 23.98\% complete, 65.09\% $L_1$ only, and 10.93\% neither under common validity. The corresponding sequential-only proportions are 31.18\%, 57.89\%, and 10.93\%.

\subsection{Sensitivity to the second-stage response mechanism}

The sensitivity experiment keeps the Gaussian full-data law and uses $\pi_2=\expit(b+0.1A+0.1Y+\gamma L_1)$. Under common validity, $\gamma=0$ and the intercept $b$ is calibrated to population complete-case proportions 0.5117, 0.2400, and 0.1161. In the identification experiment, $\gamma$ is 0.75, 1.50, or 2.25 and $b$ is calibrated separately to maintain a population complete-case proportion of 0.24. Each configuration uses 5,000 samples of size 2,500. The intercepts, empirical response shares, biases, coverage, and uncertainty in the variance reduction are reported below. Calibration uses deterministic integration; reported empirical shares need not equal the population targets exactly.

\begin{table}[htbp]
\centering
\begin{threeparttable}
\caption{Efficiency sensitivity under common validity (5,000 replicates per row).}
\label{tab:sensitivity-efficiency}
\small
\setlength{\tabcolsep}{4pt}
\begin{tabular}{@{}lrrrrr@{}}
\toprule
Missingness & Intercept & Complete & $L_1$ only & Reduction & MC interval \\
\midrule
Low & $0.19998$ & 51.16\% & 37.92\% & 1.47\% & 0.68--2.25 \\
Moderate & $-1.09906$ & 24.01\% & 65.06\% & 3.66\% & 2.52--4.85 \\
High & $-2.00020$ & 11.61\% & 77.45\% & 5.05\% & 3.81--6.29 \\
\bottomrule
\end{tabular}
\begin{tablenotes}[flushleft]\footnotesize
\item Reduction is $100\{1-\widehat{\Var}(\widehat\psi_{\seq})/\widehat{\Var}(\widehat\psi_{\cc})\}$, using centered empirical variances. The interval is a paired-bootstrap 95\% Monte Carlo interval from 10,000 resamples; it is expressed in percentage points.
\end{tablenotes}
\end{threeparttable}
\end{table}
\begin{table}[htbp]
\centering
\begin{threeparttable}
\caption{Identification sensitivity with population complete-case proportion fixed at 24\% (5,000 replicates per row).}
\label{tab:sensitivity-identification}
\small
\setlength{\tabcolsep}{4pt}
\begin{tabular}{@{}rrrrrr@{}}
\toprule
$\gamma$ & Intercept & CC bias & CC coverage & Seq. bias & Seq. coverage \\
\midrule
$0.75$ & $-1.47346$ & $-0.01363$ & 83.22\% & $0.00011$ & 95.06\% \\
$1.50$ & $-1.91220$ & $-0.02633$ & 60.88\% & $-0.00000$ & 94.90\% \\
$2.25$ & $-2.41753$ & $-0.03696$ & 44.28\% & $-0.00018$ & 95.56\% \\
\bottomrule
\end{tabular}
\begin{tablenotes}[flushleft]\footnotesize
\item CC denotes the coarsened oracle procedure using $G,Q_0,\rho$; Seq. denotes sequential augmentation. Intercepts solve the stated population response-share equation by deterministic integration.
\end{tablenotes}
\end{threeparttable}
\end{table}

\begin{figure}[htbp]
\centering
\includegraphics[width=0.98\textwidth]{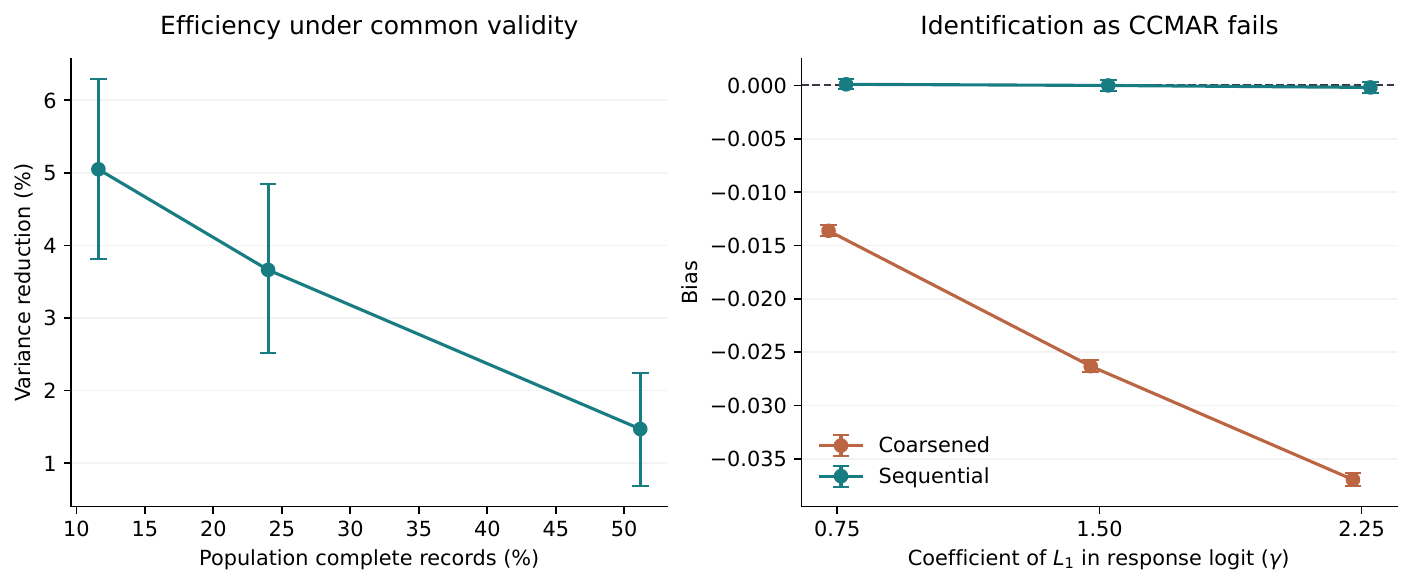}
\caption{Sensitivity summaries from the same results as \cref{tab:sensitivity-efficiency,tab:sensitivity-identification}. Left: centered empirical variance reduction and paired-bootstrap Monte Carlo intervals. Right: bias and Monte Carlo intervals computed as the estimated bias plus or minus 1.96 Monte Carlo standard errors.}
\label{fig:sensitivity}
\end{figure}

\FloatBarrier
\subsection{A design with bounded treatment probabilities}\label{sec:bounded}

To assess estimated-nuisance behavior under uniform positivity, we also specify a prospective full-data law. Independently, $L_1\sim\mathrm{Bernoulli}(0.5)$ and $L_2\sim\mathrm{Uniform}(-1,1)$. Let
\begin{align*}
A\mid L_1,L_2&\sim\mathrm{Bernoulli}\{\expit(-0.3+0.6L_1+0.7L_2)\},\\
\mu_a(L_1,L_2)&=0.25+0.15a+0.10L_1+0.05L_2,\\
Y\mid A,L_1,L_2&\sim\mathrm{Beta}\{8\mu_A,8(1-\mu_A)\}.
\end{align*}
Potential outcomes may be generated from these conditional beta laws independently of treatment given the confounders, ensuring conditional exchangeability. The ATE is exactly $0.15$, and $\expit(-1)\le e(V)\le\expit(1)$, approximately $[0.269,0.731]$. Outcomes and response probabilities are also bounded as required for the moment and positivity conditions. The two response mechanisms are the same functions of $A,Y,L_1$ as above. This design supplies an instance of the identifying assumptions and uniform positivity; it does not by itself verify the nuisance-rate conditions.

All nuisance functions are estimated with the fixed five-fold procedure in \cref{app:crossfit}. Each configuration uses 1,000 replicates. No tuning parameter is selected from the Monte Carlo performance. The full-data benchmark uses all simulated confounders; each incomplete-data estimator receives only its observed covariates.

\begin{table}[htbp]
\centering
\begin{threeparttable}
\caption{Estimated-nuisance performance with bounded treatment probabilities.}
\label{tab:bounded}
\small
\setlength{\tabcolsep}{4pt}
\begin{tabular}{@{}rlrrrrr@{}}
\toprule
$n$ & Estimator & Bias & SD & Mean SE & RMSE & Coverage \\
\midrule
\multicolumn{7}{l}{\textit{Common validity}} \\
2,500 & Full data & $-0.00003$ & $0.00661$ & $0.00665$ & $0.00661$ & 95.1\% \\
 & Coarsened & $-0.00007$ & $0.00935$ & $0.00958$ & $0.00935$ & 95.9\% \\
 & Sequential & $-0.00001$ & $0.00890$ & $0.00896$ & $0.00890$ & 96.0\% \\
\addlinespace[3pt]
5,000 & Full data & $-0.00013$ & $0.00455$ & $0.00467$ & $0.00455$ & 95.6\% \\
 & Coarsened & $-0.00019$ & $0.00607$ & $0.00618$ & $0.00607$ & 95.4\% \\
 & Sequential & $-0.00027$ & $0.00548$ & $0.00567$ & $0.00548$ & 96.4\% \\
\addlinespace[3pt]
\multicolumn{7}{l}{\textit{Sequential MAR only}} \\
2,500 & Full data & $0.00019$ & $0.00675$ & $0.00665$ & $0.00675$ & 95.1\% \\
 & Coarsened & $0.00437$ & $0.00870$ & $0.00850$ & $0.00973$ & 91.4\% \\
 & Sequential & $0.00000$ & $0.01227$ & $0.01104$ & $0.01227$ & 95.4\% \\
\addlinespace[3pt]
5,000 & Full data & $-0.00001$ & $0.00476$ & $0.00467$ & $0.00475$ & 94.0\% \\
 & Coarsened & $0.00403$ & $0.00574$ & $0.00566$ & $0.00701$ & 87.5\% \\
 & Sequential & $-0.00003$ & $0.00643$ & $0.00632$ & $0.00643$ & 95.6\% \\
\addlinespace[3pt]
\bottomrule
\end{tabular}
\begin{tablenotes}[flushleft]\footnotesize
\item SD is the empirical standard deviation; mean SE is the average estimated standard error. Coverage uses nominal 95\% normal intervals. Each configuration uses 1,000 replicates; the true ATE is 0.15. Coverage Monte Carlo SE is at most 1.58 percentage points (about 0.69 near 95\% coverage).
\end{tablenotes}
\end{threeparttable}
\end{table}
Across the four bounded-design configurations, sequential coverage ranges from 95.4\% to 96.4\%. The bias, empirical variability, and average estimated standard errors are shown together so that coverage can be assessed against its Monte Carlo uncertainty. The common-validity and sequential-only scenarios must still be distinguished when interpreting the coarsened results. A smaller bias does not guarantee a smaller finite-sample mean squared error: in the sequential-only configuration at $n=2{,}500$, sequential RMSE is $0.01227$, compared with $0.00973$ for the biased coarsened estimator. The efficiency-bound comparison does not establish a uniform finite-sample MSE ordering.

\subsection{Estimated nuisances in the stress design}\label{sec:estimated-nuisance}

The corresponding Gaussian-design results appear in \cref{app:crossfit}, including sample sizes 2,500 through 20,000 and oracle substitutions for selected nuisance layers. With all nuisances estimated, sequential coverage ranges from 83.8\% to 89.6\% across the eight sample-size and scenario configurations. The observed undercoverage is a limitation of this estimator and learner in that design; a reduction in bias or improvement with sample size does not establish the asymptotic rate conditions. The oracle-layer comparisons are diagnostics of the specified fitting procedure, rather than a general decomposition of all sources of finite-sample error.

\FloatBarrier

\section{Discussion}

The exact variance identity gives a direct interpretation of partial-record value. The loss from coarsening depends on the residual information $(Q_1-Q_0)^2$ and the factor $(1-\pi_2)/(\pi_1\pi_2)$. Holding the full-data law and $\pi_1$ fixed, decreasing $\pi_2$ increases the absolute variance gap wherever $Q_1\ne Q_0$. Weak response positivity can nevertheless increase both efficiency bounds and make finite-sample estimation unstable. Many $L_1$-only records yield no gain when $Q_1=Q_0$; when $\pi_2=1$, the gap is zero regardless of the predictive value of $L_1$.

The model comparison also clarifies assumption strength. Under \cref{ass:common}, complete-case coarsening wastes information but remains valid. If $L_1$ affects the chance that $L_2$ is observed, sequential MAR may remain plausible because that chance depends on information observed at stage two, whereas CCMAR generally fails after $L_1$ is hidden inside the all-complete indicator. In that setting, the pattern-aware estimator is not simply a more precise version of the CCMAR estimator; it is based on a different identifying model.

The present results sit between two literatures. General sequential augmentation and efficiency under monotone attrition are well established \citep{rrz1994,gill1997,barnwell2025}. Recent causal work addresses complete-case MAR \citep{levis2025}, multiple missing confounders or exposures under other MAR/MNAR restrictions \citep{wen2026}, and broad numerical comparisons of practical missing-confounder estimators \citep{benz2025,williamson2026}. The narrow value added here is the exact causal score comparison between retaining and collapsing a monotone intermediate pattern. We do not solve the nonmonotone $2^q$-pattern projection problem under the original CCMAR model, and we do not claim the generic monotone-MAR influence-function transformation as new.

The ordered-confounder construction in Section~4.2 of \citet{wen2026} already uses sequential observation probabilities and iterated projections to obtain an influence function and a targeted maximum likelihood estimator. That construction allows an incompletely observed exposure and imposes outcome-independent missingness restrictions. Here the exposure is fully observed, while both response hazards may depend on the observed treatment and outcome. Our comparison quantifies the efficiency loss from collapsing the intermediate pattern on the common-validity submodel and distinguishes that loss from bias when coarsening invalidates identification. The different identifying restrictions are essential to this comparison; an ordering of confounders alone does not distinguish the present work from the existing construction.

The two estimated-nuisance designs serve different purposes. The bounded-propensity design satisfies uniform positivity and assesses the implementation in a regular setting. The Gaussian design has a finite oracle efficiency bound but violates uniform treatment positivity; it exposes substantial sensitivity to nuisance fitting and poor interval calibration. Neither experiment establishes the product-rate conditions in \cref{thm:asymptotics}. In particular, fitting a fixed spline basis does not by itself establish the required nuisance convergence rates. All results are synthetic, and the simulation comparison uses one specified learner and two missingness mechanisms. It does not establish performance across arbitrary nuisance misspecification or relative to multiple imputation and other practical alternatives.

The variance difference in \cref{eq:variance-gap} is a first-order efficiency-bound comparison. Finite-sample nuisance error can obscure or reverse the corresponding ordering of estimator variances. The normal intervals in the stress design should therefore not be treated as reliable at the sample sizes studied. A bootstrap or a nuisance-aware variance procedure would require its own validation; the present results do not establish that either would repair the undercoverage.

Several extensions would strengthen the practical case: a real-data or plasmode application with defensible acquisition ordering; broader learner and nuisance-misspecification comparisons; and a validated inference procedure for difficult overlap settings. Monotonicity rules out $L_2$-only records and may not suit unstructured electronic health records. Neither sequential MAR nor causal exchangeability is testable from the observed data alone, so the identifying model must be justified by substantive knowledge of data collection. Extension to more monotone stages is also of interest, but is not established here.

\section*{Data and code availability}

All data used in this study are synthetic. The simulation code and results are available at
\url{https://github.com/KeivanBolouri/pattern-aware-sequential-mar}.
The original CCMAR software of \citet{levis2025} is available at
\url{https://github.com/alexlevis/flex-ate-confounders-MAR}.

\section*{Reproducibility}

Every numerical table and figure is generated from the result files in the supplement. Simulation streams are indexed by configuration and replicate; the base seed for estimated nuisances is 20260910. The README specifies all streams and software versions. The implementation accepts missing covariates as unavailable values and evaluates each augmentation only on the required observation pattern. Regression checks verify invariance to unavailable covariates, oracle response probabilities in training weights, equality with the directly evaluated oracle estimating equation, and centered variance calculations. Numerical quadrature is checked at two resolutions. These checks address specific implementation risks; they do not prove finite-sample coverage or the nuisance-rate assumptions of the asymptotic theorem.

\begingroup
\small
\begin{spacing}{1.15}
\setlength{\bibsep}{5pt}

\end{spacing}
\endgroup

\appendix
\bigskip
\section{Tangent space of the observed-data model}\label{app:tangent}

\Cref{thm:eif} verifies that $D_{\seq}$ represents the pathwise derivative of
$\psi$ along five families of score directions. That verification is exhaustive
only if those families span the tangent space of $\mathcal M$. The following
lemma supplies the missing step; it is the two-stage case of the
coarsening-at-random tangent-space decomposition
\citep{rrz1994,gill1997,tsiatis2006}, written out here because the causal
argument depends on it.

\begin{lemma}[Orthogonal decomposition of $L_2^0(P_O)$]\label{lem:tangent}
Let \cref{ass:smar} hold and suppose in addition
$\epsilon_\pi\le\pi_1(Z_0)\le1-\epsilon_\pi$ and
$\epsilon_\pi\le\pi_2(Z_1)\le1-\epsilon_\pi$ almost surely, so that all three
response patterns occur with positive probability. Then the subspaces
$\mathcal T_0,\mathcal T_{\pi_1},\mathcal T_1,\mathcal T_{\pi_2},\mathcal T_2$
defined in the proof of \cref{thm:eif} are closed, mutually orthogonal, and
\[
L_2^0(P_O)=\mathcal T_0\oplus\mathcal T_{\pi_1}\oplus\mathcal T_1
\oplus\mathcal T_{\pi_2}\oplus\mathcal T_2 .
\]
Consequently the tangent space of $\mathcal M$ at $P_O$ is all of $L_2^0(P_O)$,
and a pathwise differentiable parameter has a unique gradient.
\end{lemma}

\begin{proof}
\emph{Orthogonality.} All five statements follow from iterated conditioning
together with \cref{eq:mar1,eq:mar2}. Two are representative. For
$s_0\in\mathcal T_0$ and $R_1s_1\in\mathcal T_1$,
\[
\E(s_0R_1s_1)=\E\{s_0\,\E(R_1s_1\mid Z_0)\}
=\E\{s_0\pi_1\E(s_1\mid Z_0)\}=0,
\]
where $\E(s_1\mid Z_0,R_1=1)=\E(s_1\mid Z_0)$ by \cref{eq:mar1}. For
$R_1s_1\in\mathcal T_1$ and $R_2s_2\in\mathcal T_2$, using $R_1R_2=R_2$ and
conditioning first on $Z_1$,
\[
\E(R_1s_1R_2s_2)=\E\{s_1\,\E(R_2s_2\mid Z_1)\}
=\E\{s_1\pi_1\pi_2\E(s_2\mid Z_1)\}=0,
\]
because \cref{eq:mar2} gives $\E(C_2\mid Z_1,R_1=1,L_2)=\pi_2(Z_1)$ and
\cref{eq:mar1} gives $L_2\mid(Z_1,R_1=1)\sim L_2\mid Z_1$. The remaining pairs
are obtained the same way, using $\E(R_1-\pi_1\mid Z_0)=0$ and
$\E\{R_1(C_2-\pi_2)\mid Z_1\}=0$.

\emph{Completeness.} The observed datum is determined by $Z_0$ on $\{R_1=0\}$,
by $Z_1$ on $\{R_1=1,C_2=0\}$, and by $Z_2$ on $\{R_2=1\}$. Hence any
$f\in L_2^0(P_O)$ can be written as
\[
f=\alpha(Z_0)\1\{R_1=0\}+\beta(Z_1)\1\{R_1=1,C_2=0\}+\gamma(Z_2)\1\{R_2=1\},
\]
and the two-sided bounds make $\alpha,\beta,\gamma$ square integrable. Define
\begin{align*}
h_2&=\E(\gamma\mid Z_1)-\beta, &\qquad s_2&=\gamma-\E(\gamma\mid Z_1),\\
h_1&=\E\!\left[(1-\pi_2)\beta+\pi_2\E(\gamma\mid Z_1)\,\middle|\,Z_0\right]-\alpha,
&\qquad s_1&=\beta+\pi_2h_2-\alpha-h_1,\\
s_0&=\alpha+\pi_1h_1.
\end{align*}
By construction $\E(s_2\mid Z_1)=0$ and $\E(s_1\mid Z_0)=0$. Writing
$D=s_0+(R_1-\pi_1)h_1+R_1s_1+R_1(C_2-\pi_2)h_2+R_2s_2$ and evaluating on each
pattern gives $D=\alpha$ when $R_1=0$; $D=\alpha+h_1+s_1-\pi_2h_2=\beta$ when
$R_1=1,C_2=0$; and $D=\beta+h_2+s_2=\gamma$ when $R_2=1$. Thus $D=f$. Finally
\[
\E(s_0)=\E\{(1-\pi_1)\alpha\}+\E\{\pi_1(1-\pi_2)\beta\}+\E(\pi_1\pi_2\gamma)=\E(f)=0,
\]
the three terms being exactly the pattern probabilities implied by
\cref{ass:smar}. Uniqueness of the representation follows from orthogonality.
\end{proof}

\begin{remark}[Boundary cases]\label{rem:boundary}
Only the completeness half uses the upper bounds on $\pi_1,\pi_2$, and their
role is solely to guarantee that each of the three response patterns occurs
with positive probability, so that $\alpha,\beta,\gamma$ are identified and
square integrable.

The two boundaries behave differently, and it is worth being precise about
which subspace degenerates. Suppose $\pi_2\equiv1$, so that $L_2$ is observed
whenever $L_1$ is and $R_2=R_1$. Then $C_2-\pi_2\equiv0$, so
$\mathcal T_{\pi_2}=\{0\}$: there is no second response mechanism to score.
But $\mathcal T_2=\{R_2s_2(Z_2):\E(s_2\mid Z_1)=0\}$ is \emph{not} affected,
because the conditional law of $L_2$ given $Z_1$ is still a free component of
the model and is still observed on $\{R_1=1\}$. The pattern
$\{R_1=1,C_2=0\}$ is empty, $\beta$ drops out, and the construction in the proof
goes through with $h_2=0$, $s_2=\gamma-\E(\gamma\mid Z_1)$,
$s_1=\E(\gamma\mid Z_1)-\E(\gamma\mid Z_0)$ and $h_1=\E(\gamma\mid Z_0)-\alpha$,
giving $L_2^0(P_O)=\mathcal T_0\oplus\mathcal T_{\pi_1}\oplus\mathcal T_1
\oplus\mathcal T_2$. Symmetrically, $\pi_1\equiv1$ removes
$\mathcal T_{\pi_1}$ and leaves $\mathcal T_1$ intact. The influence function in \cref{thm:eif} remains valid. When $\pi_2\equiv1$, its two inverse-weighted increments combine:
\[
\frac{R_1}{\pi_1}(Q_1-Q_0)+\frac{R_1}{\pi_1}(G-Q_1)
=\frac{R_1}{\pi_1}(G-Q_0).
\]
The influence function is therefore $Q_0-\psi+R_1(G-Q_0)/\pi_1$; the residual increment does not vanish. When $\pi_1\equiv1$, the first two terms instead combine to $Q_1-\psi$, leaving $Q_1-\psi+R_2(G-Q_1)/\pi_2$.

When $\pi_2=1$ on a set of positive probability rather than everywhere, the same
argument applies on that set and the general one on its complement. There
\cref{eq:variance-gap} contributes nothing, since its integrand carries the
factor $1-\pi_2$; this is the degenerate boundary of \cref{cor:strict}.
\end{remark}

\section{Additional derivation of the efficiency identity}\label{app:efficiency}

This appendix makes explicit the orthogonality used in \cref{thm:gap}. Let
\[
\Delta=Q_1-Q_0,\qquad \varepsilon=G-Q_1,
\]
so $\E(\Delta\mid Z_0)=0$ and $\E(\varepsilon\mid Z_1)=0$. Under \cref{ass:common},
\begin{align*}
D_{\seq}&=Q_0-\psi+\frac{R_1}{\pi_1}\Delta
+\frac{R_1C_2}{\pi_1\pi_2}\varepsilon,\\
U:=D_{\cc}-D_{\seq}
&=\frac{R_1}{\pi_1}\left(\frac{C_2}{\pi_2}-1\right)\Delta.
\end{align*}
The product of $U$ with $Q_0-\psi$ has mean zero by conditioning on $Z_1,R_1$ and averaging over $C_2$. The product with $R_1\Delta/\pi_1$ also has mean zero for the same reason. For the last term,
\begin{align*}
\E\!\left[
U\frac{R_1C_2}{\pi_1\pi_2}\varepsilon
\right]
&=\E\!\left[
\frac{1-\pi_2}{\pi_1\pi_2}\Delta\varepsilon
\right]=0,
\end{align*}
because the coefficient and $\Delta$ are measurable with respect to $Z_1$, whereas $\E(\varepsilon\mid Z_1)=0$. Thus $U\perp D_{\seq}$. Moreover,
\begin{align*}
\E(U^2\mid Z_1)
&=\frac{\Delta^2}{\pi_1^2}
\E\!\left[R_1\left(\frac{C_2}{\pi_2}-1\right)^2\middle|Z_1\right]\\
&=\frac{1-\pi_2}{\pi_1\pi_2}\Delta^2,
\end{align*}
which establishes the identity.

\section{Observed-score calculations for response mechanisms}\label{app:scores}

For completeness, consider $S_{\pi_1}=(R_1-\pi_1)h_1(Z_0)$. Using \cref{eq:eif-centered}, each inner product $\E(D_{\seq}S_{\pi_1})$ is zero: terms without $R_1$ vanish because $\E(R_1-\pi_1\mid Z_0)=0$, while terms containing the inverse-weighted residual increments vanish after conditioning successively on $Z_0$ and $Z_1$. For stage two, write the observed score as
\[
S_{\pi_2}=(R_2-R_1\pi_2)h_2(Z_1)=R_1(C_2-\pi_2)h_2(Z_1).
\]
The $q_0$ and $q_1-q_0$ components are orthogonal because the centered Bernoulli residual has conditional mean zero. For the terminal component,
\begin{align*}
\E\!\left[
\frac{R_2}{\pi_1\pi_2}(D_{\F}-q_1)
(R_2-R_1\pi_2)h_2(Z_1)
\right]
\end{align*}
equals, after response averaging,
\[
\E\!\left[(1-\pi_2(Z_1))(D_{\F}-q_1)h_2(Z_1)\right]=0,
\]
because $\E(D_{\F}-q_1\mid Z_1)=0$. Thus the canonical gradient carries no component in either response-mechanism tangent space, as required because the causal target is a functional of the full-data law alone.

\section{Cross-fitted estimation with estimated nuisances}\label{app:crossfit}

\subsection*{Observed-data implementation}

Each sample is randomly divided into five folds. All basis transformations and fitted models used on a validation fold are trained on the other four folds. The outcome basis uses all training values of $Y$. The incomplete-data causal basis uses only $L_2$ values with $R_2=1$ in the training set. A separate full-data benchmark chooses its own basis using all training $L_2$ values. Missing covariates are represented by unavailable values: the implementation reads $L_1$ only when $R_1=1$ and $L_2$ only when $R_2=1$. There is no fitting, prediction, or knot selection using hidden simulated confounders in the incomplete-data procedures.

Continuous coordinates use cubic B-splines with six empirical quantile knots, including the boundary knots, and no spline bias column. Each fitted regression includes an intercept. The design for $Z_0$ contains the spline in $Y$, $A$, and their interactions. The design for $Z_1$ is saturated in $(A,L_1)$, including their products with the spline in $Y$. The causal design contains the spline in $L_2$, $L_1$, and their interactions. Logistic regressions use an $\ell_2$ penalty with inverse regularization parameter $C=1$; ridge regressions use penalty $10^{-3}$. Hyperparameters are fixed across designs, sample sizes, and replicates. Logistic fitting allows at most 2,000 iterations, and a convergence warning stops the run instead of being silently suppressed.

The response models are fitted on all training records for $\pi_1$ and $\rho$, and on $R_1=1$ records for $\pi_2$. Causal regressions are fitted on complete records with weights $1/(\widehat\pi_1\widehat\pi_2)$ for the sequential procedure and $1/\widehat\rho$ for the coarsened procedure. The sequential $Q_1$ regression uses complete records. Its $Q_0$ regression uses the augmented stage-two pseudo-outcome on training records with $R_1=1$. The coarsened projection regression uses complete records. These projection regressions use causal predictions fitted within the same outer training fold; there is no additional inner cross-fitting. The held-out fold is excluded from all these training operations.

Fitted treatment probabilities are truncated to $[0.01,0.99]$, and fitted response probabilities to $[0.01,1]$. Oracle substitutions are evaluated without truncation. In the stress design the true propensity can lie outside the fitted interval, so fixed truncation does not preserve the true propensity and cannot be assumed asymptotically harmless. In the bounded design the truncation thresholds lie outside the true propensity range. No truncation-sensitivity result is claimed without a corresponding run in the supplement.

\subsection*{Behavior across sample sizes}

\begin{table}[htbp]
\centering
\begin{threeparttable}
\caption{Estimated-nuisance performance in the Gaussian stress design.}
\label{tab:crossfit-n}
\small
\setlength{\tabcolsep}{4pt}
\begin{tabular}{@{}rlrrrrr@{}}
\toprule
$n$ & Estimator & Bias & SD & Mean SE & RMSE & Coverage \\
\midrule
\multicolumn{7}{l}{\textit{Common validity}} \\
2,500 & Full data & $0.00023$ & $0.01193$ & $0.00947$ & $0.01193$ & 87.0\% \\
 & Coarsened & $-0.00076$ & $0.04699$ & $0.03016$ & $0.04697$ & 84.0\% \\
 & Sequential & $-0.00135$ & $0.04827$ & $0.03087$ & $0.04826$ & 84.0\% \\
\addlinespace[3pt]
5,000 & Full data & $0.00004$ & $0.00761$ & $0.00665$ & $0.00760$ & 91.2\% \\
 & Coarsened & $0.00048$ & $0.02590$ & $0.01754$ & $0.02589$ & 85.6\% \\
 & Sequential & $0.00053$ & $0.02712$ & $0.01786$ & $0.02711$ & 85.7\% \\
\addlinespace[3pt]
10,000 & Full data & $0.00042$ & $0.00573$ & $0.00471$ & $0.00574$ & 88.0\% \\
 & Coarsened & $-0.00028$ & $0.01335$ & $0.00968$ & $0.01334$ & 85.0\% \\
 & Sequential & $-0.00033$ & $0.01325$ & $0.00965$ & $0.01324$ & 84.8\% \\
\addlinespace[3pt]
20,000 & Full data & $-0.00020$ & $0.00417$ & $0.00337$ & $0.00417$ & 88.8\% \\
 & Coarsened & $-0.00012$ & $0.00795$ & $0.00628$ & $0.00794$ & 87.6\% \\
 & Sequential & $-0.00016$ & $0.00792$ & $0.00619$ & $0.00791$ & 88.4\% \\
\addlinespace[3pt]
\multicolumn{7}{l}{\textit{Sequential MAR only}} \\
2,500 & Full data & $0.00040$ & $0.01189$ & $0.00954$ & $0.01189$ & 87.8\% \\
 & Coarsened & $-0.03198$ & $0.03702$ & $0.02398$ & $0.04891$ & 61.9\% \\
 & Sequential & $0.00104$ & $0.05048$ & $0.03204$ & $0.05047$ & 83.8\% \\
\addlinespace[3pt]
5,000 & Full data & $0.00027$ & $0.00784$ & $0.00671$ & $0.00784$ & 91.2\% \\
 & Coarsened & $-0.03238$ & $0.01761$ & $0.01308$ & $0.03686$ & 32.2\% \\
 & Sequential & $0.00076$ & $0.02027$ & $0.01546$ & $0.02027$ & 87.8\% \\
\addlinespace[3pt]
10,000 & Full data & $0.00005$ & $0.00536$ & $0.00473$ & $0.00536$ & 92.4\% \\
 & Coarsened & $-0.03283$ & $0.01057$ & $0.00843$ & $0.03449$ & 8.0\% \\
 & Sequential & $-0.00006$ & $0.01451$ & $0.01013$ & $0.01450$ & 85.8\% \\
\addlinespace[3pt]
20,000 & Full data & $0.00000$ & $0.00386$ & $0.00336$ & $0.00385$ & 90.0\% \\
 & Coarsened & $-0.03300$ & $0.00697$ & $0.00575$ & $0.03372$ & 0.0\% \\
 & Sequential & $0.00016$ & $0.00827$ & $0.00626$ & $0.00825$ & 89.6\% \\
\addlinespace[3pt]
\bottomrule
\end{tabular}
\begin{tablenotes}[flushleft]\footnotesize
\item SD is the empirical standard deviation; mean SE is the average estimated standard error. Coverage uses nominal 95\% normal intervals. Replicate counts are 1,000, 1,000, 500, and 250 at increasing sample sizes. Coverage Monte Carlo SE is at most 1.58, 1.58, 2.24, and 3.16 percentage points, respectively.
\end{tablenotes}
\end{threeparttable}
\end{table}
Under common validity, the ratio of mean estimated SE to empirical SD for the sequential estimator is 0.64, 0.66, 0.73, 0.78 at the four increasing sample sizes. These ratios assess the scale of interval miscalibration directly. The finite-sample precision comparison is given by the empirical SDs, whereas \cref{eq:variance-gap} compares first-order efficiency bounds. They need not agree at these sample sizes. Outside common validity, the coarsened estimator's bias reflects its invalid identifying model as well as nuisance fitting.

\FloatBarrier
\subsection*{Oracle substitutions and implementation checks}

\begin{table}[htbp]
\centering
\begin{threeparttable}
\caption{Oracle-substitution diagnostics in the Gaussian design ($n=2{,}500$).}
\label{tab:crossfit-layers}
\footnotesize
\setlength{\tabcolsep}{4pt}
\begin{tabular}{@{}llrrrrrr@{}}
\toprule
Scenario & Known & Seq. bias & Seq. SD & Seq. cov. & CC bias & CC SD & CC cov. \\
\midrule
Common & None & $-0.00135$ & $0.04827$ & 84.0\% & $-0.00076$ & $0.04699$ & 84.0\% \\
 & Causal & $0.00044$ & $0.01941$ & 96.0\% & $0.00077$ & $0.01877$ & 96.6\% \\
 & Response & $-0.00015$ & $0.04446$ & 82.1\% & $0.00009$ & $0.04454$ & 82.8\% \\
 & All & $0.00076$ & $0.01789$ & 96.8\% & $0.00107$ & $0.01830$ & 97.6\% \\
\addlinespace
Sequential & None & $0.00104$ & $0.05048$ & 83.8\% & $-0.03198$ & $0.03702$ & 61.9\% \\
 & Causal & $-0.00037$ & $0.01934$ & 95.5\% & $-0.03241$ & $0.01832$ & 49.4\% \\
 & Response & $0.00193$ & $0.04868$ & 82.2\% & $-0.03158$ & $0.03625$ & 59.2\% \\
 & All & $-0.00113$ & $0.01660$ & 95.0\% & $-0.03321$ & $0.01855$ & 45.6\% \\
\addlinespace
\bottomrule
\end{tabular}
\begin{tablenotes}[flushleft]\footnotesize
\item Causal: true $e,\mu_0,\mu_1$. Response: true $\pi_1,\pi_2,\rho$ throughout training and validation. All: true causal, response, and projection functions, with $Q_0^{\cc}$ for the coarsened procedure. There are 1,000 replicates per row except 500 for All. Corresponding configurations share random-number streams.
\end{tablenotes}
\end{threeparttable}
\end{table}

The causal-oracle option substitutes the true $e$ and both outcome regressions wherever $G$ is evaluated, including the training projections. The response-oracle option substitutes true $\pi_1,\pi_2,\rho$ in the causal training weights, the training pseudo-outcome for $Q_0$, and the final validation-fold augmentation. The all-oracle option also replaces the projection functions. True projection functions are projections of the true causal pseudo-outcome, so that option requires the causal-oracle option. Under sequential MAR alone, the coarsened all-oracle diagnostic uses $Q_0^{\cc}$, the complete-case conditional projection.

For the sequential procedure under common validity, the empirical SD is $0.04827$ with all nuisances estimated, $0.01941$ with the causal functions known, and $0.04446$ with the response functions known. These comparisons identify sensitivity to the specified nuisance fitting choices; they do not show that every remaining discrepancy is caused by a single layer. Shared replicate seeds make the comparisons paired, and all replicate-level results are retained.

The supplied regression checks verify invariance to hidden covariate values, oracle training weights, and equality of the all-oracle calculation with \cref{eq:seq-eif}. They also distinguish variance from MSE. These checks target implementation correctness; an all-oracle match alone neither validates estimated nuisances nor establishes nominal coverage.

\end{document}